\documentclass[a4paper,12pt]{amsart}
\usepackage{amsmath,amsthm,amsfonts,amstext,amssymb}
\usepackage{mathtools}
\usepackage{graphicx}
\usepackage{subfigure}

\newcommand{\levy}{L\'{e}vy }
\newcommand{\p}{{\mathbb P}}
\newcommand{\e}{{\mathbb E}}
\newcommand{\D}{{\mathrm d}}
\newcommand{\oX}{{\overline X}}

\newcommand{\ind}[1]{\mbox{\rm\large  1}_{\{#1\}}}

\newtheorem{theorem}{Theorem}[section]
\newtheorem{lemma}{Lemma}[section]
\newtheorem{remark}{Remark}[section]
\newtheorem{corollary}{Corollary}[section]

\begin{document}

\title[Taxation and capital injections]{Power identities for L\'{e}vy risk models under taxation and capital injections}
\author[H.~Albrecher and J. Ivanovs]{Hansj{\"o}rg Albrecher\hspace{0.1cm} and\hspace{0.1cm} Jevgenijs Ivanovs}
\address{Department of Actuarial Science, Faculty of Business and Economics, University of Lausanne}
\thanks{Supported by the Swiss National Science Foundation Project 200020-143889.}
\begin{abstract}
In this paper we study a spectrally negative L\'{e}vy process which is
refracted at its running maximum and at the same time reflected from
below at a certain level. Such a process can for instance be used to
model an insurance surplus process subject to tax payments according to
a loss-carry-forward scheme together with the flow of minimal capital
injections required to keep the surplus process non-negative. We
characterize the first passage time over an arbitrary level and the
cumulative amount of injected capital up to this time by their joint
Laplace transform, and show that it satisfies a simple power relation
to the case without refraction, generalizing results
by~\cite{albrecher_hipp_taxation} and~\cite{albrecher_renaud_taxLevy}.
It turns out that this identity can also be extended to a certain type
of refraction from below. The net present value of tax collected before
the cumulative injected capital exceeds a certain amount is determined,
and a numerical illustration is provided.
\end{abstract}

\maketitle

\section{Introduction}
The aim of this paper is to study certain power relations of level
crossing quantities for spectrally negative L\'{e}vy processes, which are
motivated by insurance applications.
Concretely, assume that the surplus process of an insurance portfolio
is modeled by a spectrally negative L\'{e}vy process, and tax payments on
profits according to
a loss-carry-forward scheme are implemented by paying a certain
proportion $\gamma$ of the premium income, whenever the surplus
process is at its running maximum.
For a constant value of $\gamma$, it was shown by~\cite
{albrecher_hipp_taxation} and~\cite{albrecher_renaud_taxLevy} that the
probability of the resulting process
to stay positive is intimately connected to the one without tax
payments by a simple power relation (see also
\cite{abbr,kyprianou_ott,kyprianou_zhou_tax} for extensions).
The implemented tax rule can alternatively be seen as a general profit
participation scheme for shareholders, which for the special case of
$\gamma=1$ reduces to
a horizontal dividend barrier strategy. Whereas in classical models
business is stopped as soon as the surplus is negative, it is natural
to consider the amount of capital
needed to bring the surplus back to zero whenever it turns negative and
henceforth continue the business operations. Under horizontal dividend
payments and
a compound Poisson model for insurance claims, this question was
considered by~\cite{dicks04}, and~\cite{kuls} showed that
it can be optimal for shareholders to ``save'' the insurance business
in this way (for another injection scheme see~\cite{nied}).

In this paper we consider capital injections below zero for the general
case $\gamma\leq1$.
This amounts to study level crossing events for a spectrally negative
L\'{e}vy process refracted at its running maximum and at the same time
reflected at zero.
We characterize the first passage time over an arbitrary level and the
cumulative amount of injected capital up to this time by their joint
Laplace transform,
and establish a simple power relation to the case without refraction.
From the proof it becomes clear that such a power identity can not
hold, if
reflection from below is generalized to refraction at the running
minimum. However, if refraction always starts at the same fixed level,
a power identity still holds.

In Section \ref{sec:refraction}, we discuss simultaneous refraction
and reflection. Section \ref{sec3} then states the main results,
which are proved in Section \ref{sec5}. In Section \ref{sec4} we
consider an application of the obtained formula to determine the net
present value of tax
collected before the cumulative injected capital exceeds an exponential
amount, and give a concrete numerical example for a compound Poisson
risk model.
Finally, in Section \ref{sec6} we illustrate with yet another example
that power identities hold in wide generality. Concretely, we use our
proof technique to extend the power tax identity for first passage times
(without capital injections) to a relaxed concept of ruin which was
considered recently in the literature.

%s2 ###
\section{Refraction and reflection}\label{sec:refraction}
For a c\`adl\`ag sample path $X_t$ of any stochastic process, consider
reflection of $X_t$ at a level~$b$ (from above) defined by
$Y_t=X_t-U_t\leq b$, where $U_t$
is a non-decreasing c\`adl\`ag function with $U_0=0\vee(X_0-b)$, whose
points of increase are contained in the set $\{t\geq0: Y_t=b\}$.
This identifies $U_t$ in a unique way, and implies that $U_t=0\vee(\oX
_t-b)$, where $\oX_t=\sup\{X_s:0\leq s\leq t\}$, see e.g.~\cite
{kella_reflection}. Essentially, $U_t$ evolves as the supremum process.

For an arbitrary $\gamma\in\mathbb R$ we call the process $X_t-\gamma
U_t$ a refraction from above, which has some interpretations in
insurance risk theory. For $\gamma=1$ we retrieve the reflected
process, which can model an insurance surplus process with dividends
paid out according to a barrier strategy with barrier~$b$, whereas
$\gamma\in(0,1)$ refers to an insurance surplus process taxed
according to a loss-carried forward scheme (see e.g. \cite
{albrecher_hipp_taxation, albrecher_renaud_taxLevy}). A value $\gamma
<0$ could refer to a model with stimulation
proportional to the increase of the maximum. Finally, the case of
$\gamma>1$ can be interpreted as inhibition, which will not be
considered further in the sequel.
In general, $\gamma$ could be allowed to depend on the current value
of $U_t$ (or on the running maximum of the refraction itself),
which leads to a more general process of the form $X_t-\int
_0^{U_t}\gamma(x)\D x$. For simplicity, we will however assume
throughout this work that $\gamma$ is a constant, and only give some
comments in Remark~\ref{rem:gammax}.

This paper focuses on processes refracted from above with rate $\gamma
\leq1$ and reflected from below.
Such a process can be defined by using one-sided refraction from above
and one-sided reflection from below locally,
and then gluing segments of paths together, see also~\cite[Sec.\
XIV.3]{APQ} where a similar procedure is used to define a two-sided reflection.
More precisely, we do the following for a given interval $[a,b]$, where
$a$ is the level for reflection from below, and $b$ is the initial
level for refraction from above.
First, we consider a free process $X_t$ until it exits $[a,b]$, at
which moment we start either reflection from below (it exits through
$a$) or refraction from above (it exits through $b$).
Assuming (w.l.o.g.) the latter, we consider the time at which the
corresponding refraction goes below $a$, and then start reflection from below.
When this reflection goes above the running maximum, the refraction
from above starts, and so on,
see Figure~\ref{fig:pic2}
for an illustration of such a process.
%
%f1 ###
\begin{figure}[t]
\centering
\includegraphics[width=3in]{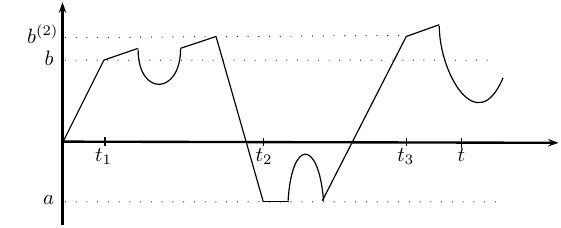}
\caption{A sample path refracted from above and reflected from below.}
\label{fig:pic2}
\end{figure}

The above procedure is described rigorously in the form of an algorithm
in the Appendix, where we also allow for two-sided refraction.
For the present model it results in a representation
%
%e1 ###
\begin{equation}\label{eq:model}
Y_t=X_t+L_t-\gamma U_t,
\end{equation}
where it is assumed that $X_0\in[a,b]$, and $\gamma\leq1$ to avoid
the case of inhibition.
Moreover, $L_t$ and $U_t$ are non-decreasing c\`adl\`ag functions,
and the points of increase of $L_t$ and $U_t$ are contained in the sets
$\{t\geq0:Y_t=a\}$ and $\{t\geq0:Y_t=\overline{Y_t}\vee b\}$ respectively.
Finally, note that $L_t$ and $U_t$ are interrelated and both depend on
the parameter $\gamma$.

%s3 ###
\section{A power identity}\label{sec3}
Throughout this work we assume that $X_t$ is a spectrally negative
L\'{e}vy process with Laplace exponent $\psi(\alpha)$ so that $\e
e^{\alpha
X_t}=e^{\psi(\alpha)t}$ for $\alpha\geq0$.
Define the first passage times
\[
\tau_y^\pm=\inf\{t\geq0:\pm X_t>y\}
\]
and recall that for all $q\geq0$ there exists a unique continuous
function $W^q:[0,\infty)\rightarrow\mathbb R_+$, such that $W^q(y)>0$
for $y>0$,
%
%e2 ###
\begin{equation}\label{eq:scale}\e_x [e^{-q\tau_y^+};\tau_y^+<\tau
_0^-]=W^q(x)/W^q(y)\text{ for }y\geq x\geq0,y>0,
\end{equation}
and $\int_0^\infty e^{-\alpha y}W^q(y)\D y=1/(\psi(\alpha)-q)$ for
$\alpha$
larger than the rightmost zero of $\psi(\alpha)-q$. This $W^q$ is called
a scale function.

For a L\'{e}vy risk model with tax, it was shown by \cite
{albrecher_renaud_taxLevy} that certain probabilities and transforms
can be related to their analogues under no taxation by power identities.
We will now generalize such power identities to the setting of
a refraction from above and reflection from below. Consider a process
$Y_t$ given by~\eqref{eq:model},
where $X_0=x>0$, the reflection barrier is placed at the level $a=0$,
and the refraction from above at rate $\gamma\leq1$ is applied
immediately, i.e. $b=x$
(it is straightforward to extend our result to $b>x$ using identities
for reflected L\'{e}vy processes).
Let also
\[
T_y=\inf\{t\geq0:Y_t>y\}
\]
be the first passage time of the refraction above the level~$y$.
\begin{theorem}\label{thm:main}
For $\gamma<1$ and $q,\theta\geq0$ it holds that
%
%e3 ###
\begin{equation}\label{eq:main}\e^{\gamma}_x e^{-q T_y-\theta
L_{T_y}}=\left(\e^0_x e^{-q T_y-\theta L_{T_y}}\right)^\frac
{1}{1-\gamma},
\end{equation}
where $y\geq x>0$ and $\e^\gamma_x$ denotes the expectation operator
for the model defined by~\eqref{eq:model} with $a=0$ and $b=x$.
\end{theorem}

It should be noted that the right hand side of~\eqref{eq:main} can be
identified using results on reflected L\'{e}vy processes. In particular,
\cite{ivanovs_levy} shows that
%
%e4 ###
\begin{equation}\label{eq:Z}\e^0_x e^{-q T_y-\theta
L_{T_y}}=Z^{q,\theta}(x)/Z^{q,\theta}(y),
\end{equation}
where $Z^{q,\theta}(x)$ is a so-called second scale function given by
\[
Z^{q,\theta}(x)=e^{\theta x}[1-(\psi(\theta)-q)\int_0^x e^{-\theta
y}W^q(y)\D y],
\]
see also~\cite{pist04} for the case when $\theta=0$.
Observe that
\[
\lim_{\theta\rightarrow\infty}\e^0_x e^{-q T_y-\theta L_{T_y}}=\e
^0_x [e^{-q T_y};L_{T_y}=0]=\e_x [e^{-q \tau^+_y};\tau^+_y<\tau
_0^-]=\frac{W^q(x)}{W^q(y)}.
\]
Similarly, for $\theta\rightarrow\infty$ the left-hand side of
\eqref{eq:main}
becomes the transform of the first passage time $T_y$ on the event that
it precedes ruin, hence we recover the tax identity (3.1) of~\cite
{albrecher_renaud_taxLevy} as a special case.

In the case $\gamma=1$ (corresponding to payments of dividends
according to a barrier strategy at the level $x$) we have $T_y=\infty$
for all $y\geq x$. Instead we look at
%
%e5 ###
\begin{equation}\label{eq:rho}
\rho_y=\inf\{t\geq0:U_t>y\},
\end{equation}
which is the first time that the amount of accumulated dividends (or
taxes) exceeds a level~$y$.
\begin{theorem}\label{thm:dividends}
For $q,\theta\geq0$ and $x>0, y\geq0$ it holds that
%
%e6 ###
\begin{equation}\label{eq:dividends}\e^1_x e^{-q \rho_y-\theta
L_{\rho_y}}=
e^{-\lambda^{q,\theta}(x) y},
\end{equation}
where
\[
\lambda^{q,\theta}(x)={Z^{q,\theta}}'(x)/Z^{q,\theta}(x)=\theta
-\frac{(\psi(\theta)-q)W^q(x)}{Z^{q,\theta}(x)}.
\]
\end{theorem}
In a somewhat different form this formula appears also in~\cite{ivanovs_levy}.
We note that for $\theta=\infty$ one has to take $\lambda
^q(x)={W^q}_+'(x)/W^q(x)$, which is intimately related to the excursion
measure, see e.g.~\cite[Lem.\ 8.2]{kyprianou_Levy}.

\begin{remark}\label{rem:exc}
The power identity \eqref{eq:main} fails to hold for a two-sided
refraction (defined in Appendix) with~$\gamma_L<1$. The case of
reflection $\gamma_L=1$ is special because in this case we know the
distance to the (lower) reflection barrier at the first passage time $T_y$
(in other words, $a^{(n)}$ in the algorithm defining the two-sided
refraction is constant, see Appendix).
%This is not the case for any other $\gamma_L<1$.

Nevertheless, from the proof in Section \ref{sec5} it becomes clear
that if one modifies the model and considers either refraction from
below always starting at a fixed level $a$ or always starting at a
fixed distance from the running maximum (rather than starting at the
current running minimum), then the power identity \eqref{eq:main} is
preserved also in the case $\gamma_L<1$.
% one modification of the refraction from below in a certain way allows
%to preserve the power identity. We alter the algorithm for a general $
%\gamma_L$ and put $a^{(n+1)}=a^{(n)}$, so that the lower barrier is
%always at the level $a$.
% This yields a model where excursions from the maximum are refracted
%from below at the level $a$ irrespective of the past minimum.
% This type of model leads to power identities too, and can be analyzed
%in a very similar way.
%\hfill$\Box$
\end{remark}
%
%s4 ###
\section{Proofs}\label{sec5}
In this section we prove Theorem~\ref{thm:main} and Theorem~\ref
{thm:dividends}.
%Note that It\^o's excursion theory is a convenient tool in the
%analysis of taxed L\'evy processes, see~\cite{kyprianou_zhou_tax}. In
%our set-up, however,
%it does not apply immediately, because of the non-Markovian structure
%of $Y_t-\overline{Y_t}$, where $Y_t=X_t+L_t-\gamma U_t$.
We construct an auxiliary process by a certain modification of paths of
the simultaneously refracted and reflected process.
This modification preserves excursions from the maximum, but leads to
the same `behavior at the maximum' as the one of the free process.
Furthermore, the auxiliary process corresponding to $\gamma=1$
exhibits a lack of memory property at its first passage times,
because the lower reflection barrier is always placed at a constant
distance from the maximum.
This gives rise to a certain exponent $\lambda(x)$, and allows to
relate this process to the processes corresponding to different $\gamma
$, see Lemma~\ref{lem:main}.
Subsequently the strong Markov property is applied to establish a
differential equation for the quantity of interest, which then yields
the results.

It is convenient to shift our process, so that $X_0=0$ and reflection
from below is applied at the level $-x<0$. Recall also that refraction
from above is applied immediately.
Note that $\e^\gamma e^{-q T_y}$ can be written as $\p^\gamma
(T_y<\infty)$ for an exponentially killed process, i.e.\ when $X_t$ is
sent to an additional absorbing state at an independent
exponentially distributed time $e_q$ with rate $q\geq0$.
The double transform $\e^\gamma e^{-q T_y-\theta L_{T_y}}$ is obtained
by additional killing at the time when $L_t$ surpasses an independent
exponentially distributed $e_\theta$.
Hence it suffices to analyze $\p^\gamma(T_y<\infty)$ for a doubly
killed process.

Let us fix some terminology and notation concerning the paths of $Y_t$.
Segments of a path of the process $Y_t-\overline Y_t$ in the
intervals when this difference is strictly negative are called
excursions of $Y_t$ (from the maximum).
The starting level of an excursion is the corresponding value of
$\overline Y_t$.
Next, consider a triplet $(Y_t,L_t,U_t)$ of paths (where each component
depends on the choice of $\gamma$) and define
\[
\tilde Y_t=X_t+L_t=Y_t+\gamma U_t.
\]
From the construction of $Y_t$ one can see that $\overline
{Y_t}=(1-\gamma)U_t$, which immediately yields $\overline{\tilde Y_t}=U_t$.
Letting
\[
\tilde T_y=\inf\{t\geq0:\tilde Y_t>y\}
\]
we see that $\tilde T_y=\rho_y$ and
for $\gamma<1$ also
%
%e7 ###
\begin{equation}\label{eq:relTilda}
\tilde T_y=T_{(1-\gamma)y}.
\end{equation}
It is noted that we could have avoided constructing the auxiliary
process, since it is possible to use the stopping time $\rho_y$
instead of $\tilde T_y$.
But then the following arguments would become less visually appealing.

%where we use the superscript $\gamma$ to emphasize the dependence on
%the parameter~$\gamma$.
% (later we also write $\tilde Y^\gamma_t$ and $L_t^\gamma$).

When $\gamma=1$ the reflecting barrier is always placed at a constant
distance $x$ from the maximum, which
together with the strong Markov property of $X_t$ implies that
%
%e8 ###
\begin{equation}\label{eq:Markov_1}
\p^1(\tilde T_{y+z}<\infty|\tilde T_y<\infty)=\p^1(\tilde
T_z<\infty)
\end{equation}
for all $y,z>0$
(note that the memoryless property of the killing times $e_q$ and
$e_\theta$ is essential here). From~\eqref{eq:Markov_1} it follows
that there exists a $\lambda(x)\geq0$ such that
%
%e9 ###
\begin{equation}\label{eq:exp}
\p^1(\tilde T_y<\infty)=e^{-\lambda(x) y},
\end{equation}
where $x$ denotes the distance between the reflecting barriers.
This provides the proof of Theorem~\ref{thm:dividends} up to the
identification of $\lambda(x)$.
%In other words, for $\gamma=1$ the all time supremum of $\tilde Y_t$
%is exponential with rate $\lambda(x)$.

\begin{lemma}\label{lem:main}
It holds for all $\gamma\leq1$ that
\[
\p^\gamma(\tilde T_{h}<\infty)=\p^1(\tilde T_{h}<\infty)+o(h)\text
{ as }h\downarrow0.
\]
\end{lemma}
\begin{proof}
In the following we will need to compare the sample paths of $\tilde
Y_t$ processes for different $\gamma$, hence throughout this proof we
write $\tilde Y^\gamma_t$ and $\tilde T^\gamma_y$
to make their dependence on $\gamma$ explicit.
For the ease of exposition, consider first the case $\gamma=0$, where
$\tilde Y^0_t$ is a process $X_t$ reflected at the level~$-x$.
Let $\delta\geq0$ be the starting level of the first excursion of
$X_t$ from the maximum exceeding height $x$; this is also the starting
level of the first excursion of $\tilde Y^1_t$
leading to reflection (i.e.\ an increase of $L^1_t$).
Note that on the event $\{\delta>h\}$ the times $\tilde T^0_{h}$ and
$\tilde T^1_{h}$ coincide.
In the following we exclusively work on the complementary event $\{
\delta\leq h\}$.

The lack of memory of $\tilde Y^1_t$ at its first passage times implies
that the number of excursions of $\tilde Y^1_t$ starting in $[0,h]$ and
leading to reflection
defines a (killed) L\'{e}vy process indexed by $h$. Hence on the event
$\{
\tilde T^1_{h}<\infty\}$ this number is Poisson distributed.
Using the lack of memory of $\tilde Y^1_t$ at $\tilde T^1_h$ we see that
\begin{align*}
&\p(\delta\leq h,\tilde T^1_h<\infty,\tilde T^1_{2h}=\infty)=\p
(\delta\leq h,\tilde T^1_h<\infty)\p(\tilde T^1_h=\infty)\\
&=O(h)(\lambda(x)h+o(h))=o(h).
\end{align*}
Hence considering $\{\delta\leq h,\tilde T^1_{h}<\infty\}$ we can
assume that $\tilde T^1_{2h}<\infty$ and also there is only one
excursion of $\tilde Y^1_t$
starting in $[0,2h]$ and leading to reflection. Comparison of the
sample paths of $\tilde Y^1_t$ and $\tilde Y^0_t$, see Figure~\ref
{fig:pic1}, reveals that $\tilde T^0_{h}<\infty$, because
the difference between them is bounded by $h$. For an arbitrary $\gamma
\leq1$ it is bounded by $(1-\gamma)h$, hence one can take
$h+(1-\gamma)h$ instead of $2h$
to finish this part of the proof.
%
%f2 ###
\begin{figure}[t]
\centering
\includegraphics[width=3in]{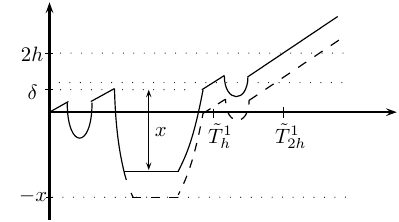}
\caption{A schematic sample path of $\tilde Y^1_t$ and $\tilde Y^0_t$
(with a dashed line).}\label{fig:pic1}
\end{figure}

Let us now consider $\{\tilde T^0_{h}<\infty\}$. Note that $\tilde
T^1_{h}=\infty$ can only happen as a consequence of killing according
to $e_\theta$.
Hence it is only required to show that this happens with probability
$o(h)$. In fact, it is enough to show that for a non-killed process
$\tilde Y_t^1$ it holds that
\[
\p^1(\delta\leq h,e_\theta\in(L_{\tilde T_\delta}-h,L_{\tilde
T_\delta}))=o(h),
\]
which follows from the independence of $e_\theta$. Again, for general
$\gamma$, $h$ in the above display is replaced by $(1-\gamma)h$.
\end{proof}

Combining Lemma~\ref{lem:main}, \eqref{eq:exp}, and \eqref
{eq:relTilda} we get for $\gamma<1$
%
%e10 ###
\begin{equation}\label{eq:oh}
\p^\gamma(T_{h}<\infty)=\p^1(\tilde T_{h/(1-\gamma)}<\infty
)+o(h)=1-\frac{\lambda(x)}{1-\gamma} h+o(h)\text{ as }h\downarrow0.\hspace*{-10pt}
\end{equation}
Let us now return to the original set-up,
where $X_0=x$ and the reflecting barrier is placed at the level 0; we
use $\p_x$ to denote the corresponding law.
\begin{proof}[Proof of Theorem~\ref{thm:main}]
Assume that $\gamma<1$ and write using the strong Markov property
\begin{equation*}
\p^\gamma_x(T_y<\infty)=\p^\gamma_x(T_{x+h}<\infty)\p^\gamma
_{x+h}(T_y<\infty).
\end{equation*}
According to \eqref{eq:oh} we have $\p^\gamma_x(T_{x+h}<\infty
)=1-\frac{\lambda(x)}{1-\gamma} h+o(h)$ as $h\downarrow0$.
Hence $\p^\gamma_{x+h}(T_y<\infty)\rightarrow\p^\gamma
_{x}(T_y<\infty)$, and moreover
%
%e11 ###
\begin{equation}\label{eq:diff}
\frac{\partial}{\partial x}\p^\gamma_{x}(T_y<\infty)=\frac
{\lambda(x)}{1-\gamma}\p^\gamma_{x}(T_y<\infty).
\end{equation}
Formally, this computation gives only the right derivative.

Let us identify $\lambda(x)$ using the existing theory. In
particular~\eqref{eq:Z} states that
$\p^0_{x}(T_y<\infty)=Z(x)/Z(y)$. Hence we obtain
$Z'(x)/Z(y)=\lambda(x)Z(x)/Z(y)$ yielding
%
%e12 ###
\begin{equation}\label{eq:lambda}
\lambda(x)=Z'(x)/Z(x)\text{ for }x>0,
\end{equation}
which also shows that $\lambda(x)$ is continuous on
$(0,\infty)$.\vadjust{\goodbreak}

It is not hard to see that for any $\gamma<1$ and fixed $y>0$ the
function $\p^\gamma_x(T_y<\infty),x\in(0,y]$ is continuous and non-zero.
Hence for all $x\in(0,y)$ we have the following right derivative:
\[
\frac{\partial}{\partial x}\ln\p^\gamma_x(T_y<\infty)=\frac
{\lambda(x)}{1-\gamma},
\]
which together with $\p^\gamma_{y}(T_y<\infty)=1$ yields
\[
\ln\p^\gamma_x(T_y<\infty)=-\frac{1}{1-\gamma}\int_x^y\lambda
(u)\D u.
\]
Uniqueness of the solution is based on the fact that a continuous
function with right derivative 0 at every point of an interval is
constant on this interval.
So we have
%
%e13 ###
\begin{equation}\label{eq:sol}\p^\gamma_{x}(T_y<\infty)=e^{-\frac
{1}{1-\gamma}\int_x^y\lambda(u)\D u},
\end{equation}
which immediately yields the power relation of Theorem~\ref{thm:main}.
\end{proof}
Finally, Theorem~\ref{thm:dividends} is a direct consequence of~\eqref
{eq:exp} and~\eqref{eq:lambda}.

\begin{remark}\label{rem:gammax}
When the refraction rate $\gamma(x)$ depends on the level, assuming
some regularity conditions (e.g. $\gamma(x)$ is continuous and bounded
away from 1),
one can still apply Lemma~\ref{lem:main} to derive the differential
equation~\eqref{eq:diff}. In this case the solution takes the form
\[
\p^\gamma_{x}(T_y<\infty)=e^{-\int_x^y\lambda(u)/(1-\gamma(u))\D u}.
\]
%
%\hfill$\Box$
\end{remark}

%s5 ###
\section{An application: Profit participation and capital
injection}\label{sec4}
As an application of Theorem~\ref{thm:main}, interpret $Y_t$ in \eqref
{eq:model} as an insurance surplus process at time $t$,
where $\gamma U_t$ is a profit participation scheme for an investor (a
proportion $\gamma$ of the profits is paid out to the investor) and,
in turn,
if needed the investor injects a minimal flow of capital into the
company to prevent its bankruptcy, i.e.\ to keep the surplus non-negative,
with $L_t$ being the total amount injected up to time $t$.
Alternatively, one can think of $\gamma U_t$ as tax payments for
profits up to time $t$ according to a loss-carry forward scheme with
constant tax rate $0<\gamma<1$
(cf. \cite{albrecher_hipp_taxation}) and $L_t$ would then be the
necessary amount of capital up to time $t$ to bail out the insurance
company to prevent bankruptcy.
Consider an upper limit~$e_\theta$ for the cumulative amount that the
investor is willing to inject,
which is assumed to be an independent exponential random variable with
rate parameter $\theta\geq0$ (it can be interpreted as impatience of
the investor).
Whenever this limit is exceeded the company is not bailed out anymore
and has to go out of business. Put differently, for each infinitesimal
required injection $h$, the investor will stop payments with
probability $\theta h$ (independently of everything else). This concept
extends the notion of classical ruin (which is retrieved for $\theta
=\infty$), and leads to an interesting trade-off between collected
profits (or tax) and injected capital.

The expected discounted profit (tax) payments for this model can be
written as
\[
V(\gamma)=\frac{\gamma}{1-\gamma}\e^{\gamma}_x\int_0^\infty
e^{-qt}\ind{L_t<e_\theta}\D\overline Y_t,\quad\gamma<1,
\]
where $q> 0$ is the discount rate. Note that each $\D\overline Y_t=\D
y$ corresponds to $\gamma/(1-\gamma)\D y$ tax payment.
Recalling that $\overline Y_t$ is continuous, and using a standard
change of variable argument with $\overline Y_t=y$ and $t=T_y$ we obtain
\begin{align}\label{eq:V}\quad
V(\gamma)&=\frac{\gamma}{1-\gamma}\e^{\gamma}_x\int
_x^\infty e^{-qT_y}\ind{L_{T_y}<e_\theta}\D y=
\frac{\gamma}{1-\gamma}\int_x^\infty\e^{\gamma}_x
[e^{-qT_y-\theta L_{T_y}}]\D y\nonumber\\[-8pt]\\[-8pt]
&=\frac{\gamma}{1-\gamma}\int_x^\infty\left(\frac{Z^{q,\theta
}(x)}{Z^{q,\theta}(y)}\right)^\frac{1}{1-\gamma}\D y,\nonumber
\end{align}
where in the second step we use Fubini's theorem and the independence
of $e_\theta$, and in the last step we invoke Theorem~\ref{thm:main}.
This formula is an extension of Equation (3.2) of \cite
{albrecher_renaud_taxLevy}, which is retained for $\theta\to\infty$
(the case without capital injections).

If we choose $\gamma=1$ (in which case the profit participation
reduces to a horizontal dividend barrier strategy), then we get in a
similar way by using Theorem~\ref{thm:dividends} that the expected
discounted dividends $V(1)$ are given by
\begin{align}
V(1)&=\e_x^1\int_0^\infty e^{-qt}\ind{L_t<e_\theta}\D U_t=\int
_0^\infty\e^1_x [e^{-q\rho_y-\theta L_{\rho_y}}]\D y\nonumber\\[-8pt]\\[-8pt]
&=\frac{1}{\lambda^{q,\theta}(x)}=Z^{q,\theta}(x)/{Z^{q,\theta}}'(x).\nonumber
\end{align}
As above, for $\theta\to\infty$ we get back to the classical formula
without capital injections, where $Z$ is replaced by $W$ (see e.g.
Equation (3) in \cite{renz}).

The quantity $V(\gamma)$ can consequently be computed explicitly
whenever the function~$Z$ has an explicit representation.
This is for instance the case for a Poisson stream of phase-type claims
(for a detailed discussion of explicit cases cf.~\cite
{scale_functions_examples}).

%s5.1 ###
\subsection{A numerical example}
Let us consider a concrete simple example, for which the scale function
$W(x)$ has an explicit form, and hence
the expected discounted profit (tax) payments $V(\gamma)$ as
identified in~\eqref{eq:V} can be easily evaluated.
We assume that the driving process is a Cram\`er-Lundberg risk process
$X_t=x+ct- \sum_{n=1}^{N(t)}M_i$, where $N(t)$ is a homogeneous
Poisson process with rate 1, the insurance claims $M_i$ are independent
and identically distributed exponential random variables with mean $m$
and the constant premium intensity is chosen as $c=1$, so that the
drift of $X$ is then given by $\e X(1)=1-m$. Choose further the initial
capital $x=1$, the discount factor $q=0.01$, and the investor
impatience parameter $\theta=1$.

Figure~\ref{fig:plots1} depicts $V$ as a function of $\gamma$ for
different values of the drift. Essentially, the shape of these functions
is the same as in the case of classical ruin ($\theta=\infty$), but
higher in absolute value due to the longer life-time of the process.
This shape reflects that overly large values of $\gamma$ may lead to
an early ruin resulting in a smaller profit.

%f3 ###
\begin{figure}[t]
\centering
\includegraphics[width=0.4\textwidth]{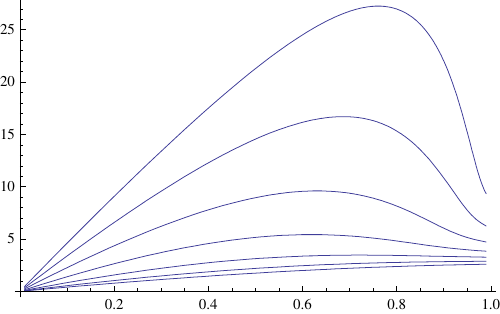}
\caption{$V(\gamma)$ for drift $=(0.5,0.4,0.3,0.2,0.1,0,-0.1)$;
from top to bottom.}\label{fig:plots1}
\end{figure}

%f4 ###
\begin{figure}[t]
\centering
\subfigure[$V^1(\gamma)$ (thick) and $V^\infty(\gamma)$]{\includegraphics[width=0.4\textwidth]{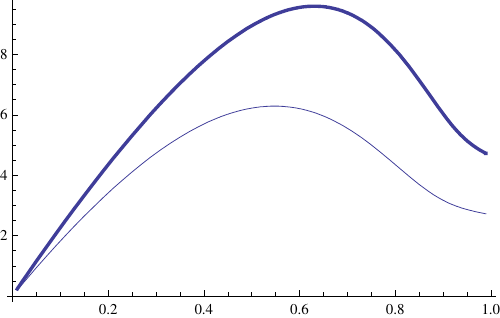}\label{fig:plots2a}}
\quad
\subfigure[$V^1(\gamma)-V^\infty(\gamma)$]{\includegraphics[width=0.4\textwidth]{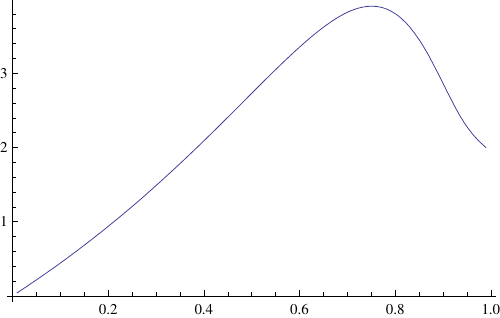}
\label{fig:plots2b}}
\caption{$V(\gamma)$ for $\theta=1$ and $\theta=\infty$.}
\end{figure}

In Figure~\ref{fig:plots2a}, this is visualized by comparing $V(\gamma
)$ for $\theta=1$ and $\theta=\infty$ for a fixed drift of $\e
X(1)=0.3$, and Figure~\ref{fig:plots2b}
depicts the increase of $V(\gamma)$ as compared to the case of
classical ruin. This expected increase of profit comes at the cost of
the capital injections,
whose expected value does not exceed $\e e_\theta=1$. The latter is in
fact a crude upper bound,
because of two reasons: no discounting, and the fact that cumulative
injections may never reach the threshold $e_\theta$.
These results show that on average it can be quite advantageous for an
investor to perform these capital injections,
in particular for those $\gamma$ for which the difference $V^1(\gamma
)-V^\infty(\gamma)$ is larger than~1.
If one would compare this difference to the actual expected discounted
investments, the effect would be even more pronounced.
The analysis of the net present value of injections is, however,
considerably more involved, and could be an interesting direction for
future work.

%s6 ###
\section{Power identities under a relaxed ruin concept}\label{sec6}
It turns out that power relations similar to \eqref{eq:main} hold in
quite wide generality.
Essentially, it is only required that killing and modification (such as
reflection) of excursions of the (non-taxed) process is
done in a memoryless way (in other words, what happens after the first
passage time $T^0_y$ is independent from the past and has the same law
as the original process started in $y$).
Of course, one still has to handle model-specific technical details
similar to those contained in Lemma~\ref{lem:main}.

For illustration, let us consider an example from~\cite{AGS} and~\cite
{volkmar}, where bankruptcy is declared at some rate $\theta>0$ when
the risk process is below zero (there is no reflection from below).
In other words, the killing occurs when the cumulative time $X_t$ spent
below zero surpasses an independent exponential random variable
$e_\theta$ (one can also introduce dependence of $\theta$ on the
level, but for clarity we refrain from doing so, and only note that
generalizations of power identities
to arbitrary measurable, locally bounded functions $\theta(x)$ do not
cause additional problems). As before we assume that $X_t$ is a
spectrally negative L\'{e}vy process (no reflection from below).
The concept of occupation times plays an important role in this
setting. Let
\[
M(A,t)=\int_0^t\ind{X_s\in A}\D s
\]
be the time $X$ spends in a Borel set $A$ up to time $t$.
\begin{theorem}\label{thm:last}Consider the model~\eqref{eq:model}
without reflection from below ($a=-\infty, b=x\geq0$), and let $\nu
_\theta$ be the time of bankruptcy:
\[
\nu_\theta=\inf\{t\geq0:M((-\infty,0),t)>e_\theta\}.
\]
Then for all $\gamma<1$ and $q\geq0$ it holds that
\[
\e^{\gamma}_x [e^{-q T_y};T_y<\nu_\theta]=\left(\e^0_x [e^{-q
T_y};T_y<\nu_\theta]\right)^\frac{1}{1-\gamma}.
\]
\end{theorem}
\begin{proof}%[Proof sketch]
Without real loss of generality one can assume that $q=0$.
One can repeat the arguments from the previous section. In fact, many
things simplify since there is no process $L_t$.
In particular, paths of the processes $\tilde Y^\gamma_t$ (and $X_t$)
are the same, but the intervals of times
when the processes are in danger of bankruptcy are different for
different~$\gamma$, and so the killing points are different.
In order to (re-)establish Lemma~\ref{lem:main}, we have to show that
the differences between `in danger' sets up to the time $\tau_h^+$ are
small in certain sense.
It is enough to show that
%
%e14 ###
\begin{equation}\label{eq:M}\p(M([-x+\gamma h,-x+h),\tau
_h^+)>e_\theta)=o(h)
\end{equation}
as $h\downarrow0$. Moreover, to establish the differential equation
\eqref{eq:diff} we have to show (for the reason of continuity) that
%
%e15 ###
\begin{equation}\label{eq:M2}M(\{x\},t)=0 \text{ a.s.\ for any }t,x.
\end{equation}
The latter fact is well-known, see~\cite[Prop.~I.15]{bertoin}.
So it is only left to show that \eqref{eq:M} holds.

%For simplicity of notation assume that $\gamma=0$ (so we are comparing
%the processes $\tilde Y^0_t$ and $\tilde Y^1_t$).
The probability in \eqref{eq:M} can be bounded from above by
\[
\p(\tau_{x-h}^-<\tau_h^+)\p(M([-(1-\gamma)h,(1-\gamma)h],\tau
^+_{x+(1-\gamma)h})>e_\theta).
\]
In short, the process must go below the upper boundary of the interval,
then we start it at the lower boundary and make the strip twice as
large, so that it starts in the middle.
The first probability is given by $1-W(x-h)/W(x)=W'_-(x)/W(x)h+o(h)$,
and the second decreases to 0 as $h\downarrow0$, because
$M([-h,h],\tau_y^+)\rightarrow0$ for any $y>0$ a.s. (use \eqref
{eq:M2} and the fact that either $\oX_t\rightarrow\infty$ a.s.\ or
$X_t\rightarrow-\infty$ a.s.).
This concludes the proof.
\end{proof}

\begin{corollary}
For the model of Theorem~\ref{thm:last} and $q\geq0$ it holds that
\begin{align*}
&\e^{\gamma}_x [e^{-q T_y};T_y<\nu_\theta]=\left(\frac
{Z^{q,\Phi}(x)}{Z^{q,\Phi}(y)}\right)^{\frac{1}{1-\gamma}},
&\gamma<1,\\
&\e^{\gamma}_x [e^{-q \rho_y};\rho_y<\nu_\theta]=\exp\left
(-\frac{{Z^{q,\Phi}}'(x)}{Z^{q,\Phi}(x)}y\right), &\gamma=1,
\end{align*}
where $\Phi$ is the unique positive solution of $\phi(\Phi)=q+\theta$.
\end{corollary}
\begin{proof}
It holds that
\[
\e_x [e^{-q \tau_y^+};\tau_y^+<\nu_\theta]=Z^{q,\Phi
}(x)/Z^{q,\Phi}(y),
\]
which can be easily deduced from the results by~\cite
{loeffen_occupation} or~\cite{albrecher2013risk}.
The rest follows from Theorem~\ref{thm:last} and its proof which
employs the ideas of Section~\ref{sec5}.
\end{proof}

\begin{appendix}
\section*{Appendix}
In the following we present an algorithm defining a two-sided
refraction of a c\`adl\`ag sample path $X_t$ corresponding to the
interval $[a,b]$.
It is assumed that $X_0\in[a,b]$, and $\gamma_L,\gamma_U\leq1$ to
avoid the case of inhibition.
The triplet of processes $(Y_t,L_t,U_t)$ is defined iteratively as
follows (cf. Figure~\ref{fig:pic2} depicting refraction from above
at~$b$ and reflection from below at~$a$).\\

{\bf Algorithm}:
\begin{itemize}
\item[]{\bf Initialization} ($n=0$):
$Y^{(0)}_t=X_t,U^{(0)}_t=0,L^{(0)}_t=0,t_0=0$ and
$a^{(1)}=a,b^{(1)}=b,$
\[
t_1=\inf\{t\geq0:X_t\notin[a,b]\}.
\]
\item[]{\bf Step} ($n=n+1$):
$X_t^{(n)}=Y^{(n-1)}_{t_n}+X_{t_n+t}-X_{t_n}$ for $t\geq0$.
\begin{itemize}
\item[]{\bf If} $X^{(n)}_0\geq b^{(n)}$: $L^{(n)}_t=0$ and
$Y_t^{(n)}=X_t^{(n)}-\gamma_U U_t^{(n)}$ is the refraction of
$X_t^{(n)},t\geq0$ from above at the level $b^{(n)}$. Put
\[
\Delta_n=\inf\{t\geq0:Y_t^{(n)}<a^{(n)}\}
\]
and $t_{n+1}=t_n+\Delta_n, a^{(n+1)}=a^{(n)},b^{(n+1)}=\overline
{Y^{(n)}_{\Delta_n}}$.
\item[]{\bf If} $X^{(n)}_0\leq a^{(n)}$: $U^{(n)}_t=0$ and
$Y_t^{(n)}=X_t^{(n)}+\gamma_L L_t^{(n)}$ is the refraction of
$X_t^{(n)},t\geq0$ from below at the level $a^{(n)}$. Put
\[
\Delta_n=\inf\{t\geq0:Y_t^{(n)}>b^{(n)}\}
\]
and $t_{n+1}=t_n+\Delta_n, a^{(n+1)}=\underline{Y^{(n)}_{\Delta
_n}},b^{(n+1)}=b^{(n)}$.
\end{itemize}
\end{itemize}
Finally, we set
\[
Y_t=Y^{(n)}_{t-t_n},L_t=\sum_{i=0}^{n-1}L^{(i)}_{\Delta
_i}+L^{(n)}_{t-t_n},U_t=\sum_{i=0}^{n-1}U^{(i)}_{\Delta
_i}+U^{(n)}_{t-t_n} \text{ for }t\in[t_n,t_{n+1}).
\]
Observe that the above procedure defines the process $Y_t$ for all
$t\geq0$, i.e. $t_n\rightarrow\infty$ as $n\rightarrow\infty$, because
a c\`adl\`ag function can not cross the interval $[a,b]$ infinitely
many times in finite time; here we use the fact that the intervals
$[a^{(n)},b^{(n)}]$ are increasing.
Careful examination of the above algorithm (together with known
properties of a one-sided refraction) shows that
\[
Y_t=X_t+\gamma_L L_t-\gamma_U U_t,
\]
where $L_t$ and $U_t$ are non-decreasing c\`adl\`ag functions.
Moreover, the points of increase of $L_t$ and $U_t$ are contained in
the sets $\{t\geq0:Y_t=\underline{Y_t}\wedge a\}$ and $\{t\geq
0:Y_t=\overline{Y_t}\vee b\}$ respectively.
It may be interesting to find an explicit representation of the
two-sided refraction similar to those given by~\cite{andersen_mandjes}
and~\cite{ramanan_skorohod0a} for the two-sided reflection.
\end{appendix}

\bibliographystyle{abbrv}
\bibliography{./samos}
\end{document}